\newcommand{\longversion}[1]{}
\newcommand{\shortversion}[1]{#1}

\documentclass[runningheads]{llncs}
\usepackage[T1]{fontenc}
\usepackage{graphicx}
\usepackage{hyperref}
\usepackage{amssymb,amsmath}
\usepackage{color}
\usepackage{makecell}
\usepackage{mathtools}
\usepackage{bbm}
\usepackage{setspace}

\DeclareMathOperator*{\argmin}{arg\,min}

\usepackage[usenames,svgnames,dvipsnames]{xcolor}

\usepackage{tikz}
\usetikzlibrary{arrows,positioning}

\newdimen\prevdp
\def\leftlabel#1{\noalign{\prevdp=\prevdepth
   \kern-\prevdp\nointerlineskip\vbox to0pt{\vss\hbox{\ensuremath{#1}}}\kern\prevdp}}

\usepackage{url}

\usepackage{enumerate}

\usepackage{xspace}
\newcommand{\NP}{\ensuremath{\mathsf{NP}}\xspace}
\newcommand{\NPC}{\ensuremath{\mathsf{NP}}\text{-complete}\xspace}
\newcommand{\NPH}{\ensuremath{\mathsf{NP}}\text{-hard}\xspace}
\newcommand{\PNPH}{para-\ensuremath{\mathsf{NP}\text{-hard}}\xspace}
\newcommand{\el}{\ensuremath{\ell}\xspace}

\newcommand{\WOH}{\ensuremath{\mathsf{W[1]}}-hard\xspace}

\newcommand{\FPT}{\ensuremath{\mathsf{FPT}}\xspace}
\newcommand{\ETH}{\ensuremath{\mathsf{ETH}}\xspace}
\newcommand{\OPT}{\ensuremath{\mathsf{OPT}}\xspace}

\newcommand{\Pb}{\ensuremath{\mathsf{P}}\xspace}
\newcommand{\tw}{\text{tw}\xspace}

\let\oldlambda\lambda
\renewcommand{\lambda}{\ensuremath{\oldlambda}\xspace}
\let\oldalpha\alpha
\renewcommand{\alpha}{\ensuremath{\oldalpha}\xspace}
\let\oldDelta\Delta
\renewcommand{\Delta}{\ensuremath{\oldDelta}\xspace}

\newcommand{\yes}{{\sc yes}\xspace}
\newcommand{\no}{{\sc no}\xspace}
\newcommand{\true}{\text{{\sc true}}\xspace}
\newcommand{\false}{\text{{\sc false}}\xspace}

\newcommand{\pa}{{\sc Connected Knapsack}\xspace}
\newcommand{\conknapsack}{\pa}
\newcommand{\pathknapsack}{{\sc Path Knapsack}\xspace}
\newcommand{\shortestpathknapsack}{{\sc Shortest Path Knapsack}\xspace}

\newcommand{\EE}{\ensuremath{\mathcal E}\xspace}

\newcommand{\GG}{\ensuremath{\mathcal G}\xspace}

\newcommand{\II}{\ensuremath{\mathcal I}\xspace}

\newcommand{\OO}{\ensuremath{\mathcal O}\xspace}
\newcommand{\PP}{\ensuremath{\mathcal P}\xspace}
\newcommand{\QQ}{\ensuremath{\mathcal Q}\xspace}
\newcommand{\RR}{\ensuremath{\mathcal R}\xspace}

\newcommand{\UU}{\ensuremath{\mathcal U}\xspace}
\newcommand{\VV}{\ensuremath{\mathcal V}\xspace}
\newcommand{\WW}{\ensuremath{\mathcal W}\xspace}
\newcommand{\XX}{\ensuremath{\mathcal X}\xspace}

\newcommand{\NB}{\ensuremath{\mathbb N}\xspace}
\newcommand{\RB}{\ensuremath{\mathbb R^+}\xspace}

\usepackage{nicefrac}

\newtheorem{observation}{\bf Observation}

\newcommand{\eps}{\ensuremath{\varepsilon}\xspace}
\renewcommand{\epsilon}{\eps}

\newcommand{\ignore}[1]{}

\newcommand{\pr}{\ensuremath{\prime}}
\newcommand{\prr}{\ensuremath{{\prime\prime}}}

\renewcommand{\leq}{\leqslant}
\renewcommand{\geq}{\geqslant}
\renewcommand{\ge}{\geqslant}
\renewcommand{\le}{\leqslant}

\usepackage{enumitem}
\setlist[enumerate]{labelwidth=!, labelindent=0pt}

\newcommand{\kp}{\text{\sc Knapsack}\xspace}
\newcommand{\vc}{\text{\sc Vertex Cover}\xspace}
\newcommand{\pvc}{\text{\sc Partial Vertex Cover}\xspace}
\newcommand{\hp}{\text{\sc Hamiltonian Path}\xspace}

\usepackage{rotating}
\usepackage{amssymb}
\usepackage{latexsym}
\usepackage{epsfig}
\usepackage{xcolor}
\usepackage{url}
\usepackage{fancyhdr}

\usepackage{caption}
\usepackage{subcaption}
\usepackage{array}
\usepackage{multirow}
\usepackage{enumerate}
\usepackage{pdflscape}
\usepackage{amsmath,graphicx,algorithm,algpseudocode,amsfonts}
\usepackage{epstopdf}
\usepackage{float}
\floatstyle{ruled}
\newfloat{Algorithms}{!thb}{lop}
\floatname{Algorithms}{Algorithm}

\allowdisplaybreaks[1]

\algnewcommand\algorithmicinput{\textbf{Input:}}
\algnewcommand\INPUT{\item[\algorithmicinput]}
\algnewcommand\algorithmicoutput{\textbf{Output:}}
\algnewcommand\OUTPUT{\item[\algorithmicoutput]}
\algnewcommand{\LineComment}[1]{\State \(\triangleright\) #1}

\usepackage{pifont}

\usepackage{mathtools}

\usepackage{algorithm}
\usepackage{algpseudocode}

\usepackage{comment}

\renewenvironment{proof}{\paragraph{Proof:}}{\hfill$\square$}

\usepackage{cleveref}

\crefname{theorem}{Theorem}{\bf Theorems}
\crefname{observation}{Observation}{\bf Observations}
\crefname{lemma}{Lemma}{\bf Lemmata}
\crefname{corollary}{Corollary}{\bf Corollaries}
\crefname{proposition}{Proposition}{\bf Propositions}
\crefname{definition}{Definition}{\bf Definitions}
\crefname{claim}{Claim}{\bf Claims}
\crefname{reductionrule}{Reduction rule}{\bf Reduction rules}
\usepackage{color,soul} % For highLight
\begin{document}

\title{Knapsack: Connectedness, Path, and Shortest-Path}
% \author{}
\author{Palash Dey\inst{1} \and
Sudeshna Kolay\inst{2} \and
Sipra Singh\inst{3}}
%
% \authorrunning{}
\authorrunning{Dey et al.}
% \institute{\email{}}
\institute{Indian Institute of Technology Kharagpur
\email{\{palash.dey\inst{1},skolay\inst{2}\}@cse.iitkgp.ac.in,sipra.singh@iitkgp.ac.in\inst{3}}}
%

% \email{lncs@springer.com}}

\maketitle              % typeset the header of the contribution

\begin{abstract}
We study the \kp problem with graph-theoretic constraints. That is, there exists a graph structure on the input set of items of \kp and the solution also needs to satisfy certain graph theoretic properties on top of the \kp constraints. In particular, we study \pa where the solution must be a connected subset of items which has maximum value and satisfies the size constraint of the knapsack. We show that this problem is strongly \NPC even for graphs of maximum degree four and \NPC even for star graphs. On the other hand, we develop an algorithm running in time $\OO\left(2^{\tw\log\tw}\cdot\text{poly}(n)\min\{s^2,d^2\}\right)$  where $\tw,s,d,n$ are respectively treewidth of the graph, the size of the knapsack, the target value of the knapsack, and the number of items. We also exhibit a $(1-\eps)$ factor approximation algorithm running in time $\OO\left(2^{\tw\log\tw}\cdot\text{poly}(n,1/\eps)\right)$ for every $\eps>0$. We show similar results for \pathknapsack and \shortestpathknapsack, where the solution must also induce a path and shortest path, respectively. Our results suggest that \pa is computationally the hardest, followed by \pathknapsack and then \shortestpathknapsack.

\end{abstract}

\section{Introduction}

The \kp problem is one of the most well-studied problems in computer science~\cite{CacchianiILM22,cacchiani2022knapsack,kellerer2004multidimensional,martello1990knapsack}. Here, we are given a set of $n$ items with corresponding sizes $w_1,\ldots,w_n$ and values $\alpha_1,\ldots,\alpha_n$, the size (aka budget) $s$ of the bag, and the goal is to compute a subset of items whose total value is as high as possible and total size is at most $s$. 

Often there exist natural graph structures on the set of items in the \kp problem, and we want the solution to satisfy some graph theoretic constraint also. For example, suppose a company wants to lease a set of mobile towers, each of which comes at a specific cost with a specific number of customers. The company has a budget and wants to serve the maximum number of customers. However, every mobile tower is in the range of some but not all the other towers. The company wants to provide fast connection between every pair of its customers, which is possible only if they are connected via other intermediary towers of the company. Here we have a natural graph structure --- the vertices are the towers, and we have an edge between two vertices if they are in the range of one another. In this situation, the company wants to lease a connected subset of towers which has the maximum total number of customers (to maximize its earning) subject to its budget constraint. We call this problem \pa.

Now, let us consider another application scenario where there is a railway network and a company wants to lease a railway path between station A and station B, maybe via other stations, to operate train services between them. Suppose the cost model puts a price tag on every station, and the links between two stations are complementary if the company owns both the stations. Each station also allows the owner to earn revenues from advertisement, etc. In this scenario, the company would like to lease from the entire railway network a path between stations A and B whose total revenues are as high as possible, subject to a budget constraint. We call this problem \pathknapsack. If the company's primary goal is to provide the fastest connectivity between A and B, then the company wants to lease a shortest path between A and B, whose total revenues are as high as possible, subject to a budget constraint. We study this problem also under the name \shortestpathknapsack. The formal definitions of the above problems are in Section~\ref{sec:prelim}.

\subsection{Contribution}

We study the computational complexity of \pa, \pathknapsack, and \shortestpathknapsack under the lens of parameterized complexity as well as approximation algorithms. We consider the treewidth of the underlying graph and the vertex cover size of the solution as our parameters. We note that both the parameters are well known to be small for many graphs appearing in various real-world applications. We summarize our results in \Cref{tbl:summary}.

\begin{table}
	\begin{center}
		\begin{tabular}{|c|c|}\hline
  
		\scriptsize \conknapsack & \makecell{\scriptsize Strongly \NPC even if max degree is $4$ [\Cref{thm:pa-gen-npc}]\\\scriptsize \NPC even for stars [\Cref{thm:pa-star-npc}]\\\scriptsize $2^{\OO(\tw\log\tw)}\cdot\text{poly}(n)\min\{s^2,d^2\}^{\dagger^\star}$ [\Cref{thm:treewidth-pa}]\\\scriptsize $2^{\OO(\tw\log\tw)}\cdot\text{poly}(n,1/\eps)^\dagger$ time, $(1-\eps)$ approximation [\Cref{thm-fptas}]\\\scriptsize No $\OO(f(vcs).\text{poly}(n,s,d)$ algorithm unless \ETH fails$^\ddagger$ 
[\Cref{thm:vcs-woh}].} \\\hline

		\scriptsize \pathknapsack & \makecell{\scriptsize Strongly \NPC even if max degree is $3$ [\Cref{thm:path-gen-npc}]\\\scriptsize  \NPC for graphs with pathwidth 2 [\Cref{thm:pk-pathwidth}]\\\scriptsize Polynomial-time algorithm for trees [\Cref{obs:pk-tree-poly}]\\\scriptsize $2^{\OO(\tw\log\tw)}\cdot\text{poly}(\min\{s^2,d^2\})^{\dagger^\star}$ [\Cref{thm:treewidth-path}]\\\scriptsize $2^{\OO(\tw\log\tw)}\cdot\text{poly}(n,1/\eps)^\dagger$ time, $(1-\eps)$ approximation [\Cref{thm-fptas-path}]\\\scriptsize $\OO\left((2e)^{2vcs} vcs^{\OO(\log vcs)}n^{\OO(1)}\right)^\ddagger$ time algorithm [\Cref{cor:vcs-path}].} \\\hline
   
			 \scriptsize \shortestpathknapsack & \makecell{\scriptsize \NPC for graphs with pathwidth 2 [\Cref{cor:shortest-pathwidth}]\\\scriptsize $\OO((m+n\log n)\cdot\min\{s^2,(\alpha(\VV))^2\})$ [\Cref{thm:shortest-algo}]\\ \scriptsize Polynomial-time algorithm for trees [\Cref{obs:shortest-tree-poly}]\\\scriptsize $(1-\eps)$ approximation in $\text{poly}(n,1/\eps)$ time [\Cref{thm-fptas-shortest-path}].}\\\hline
    
			 % \indpsetknapsack & \makecell{Strongly \NPC [\Cref{thm:ind}]\\$\OO\left(2^{\tw}\cdot\text{poly}(\min\{s,d\})\right)^{\dagger^\star}$ [\Cref{thm:treewidth-indset}]\\$\OO\left(2^{\tw}\cdot\text{poly}(n,1/\eps)\right)^\dagger$ time, $(1-\eps)$ approximation [\Cref{thm-fptas-indset}]\\\WOH parameterized by $k^{\star\star}$ [\Cref{obs:ind-woh}].} \\\hline
    
		\end{tabular}
  \caption{Summary of results. $\dagger: \tw$ is the treewidth of the input graph; $\star: s$ and $d$ are respectively target size and target profit; $\ddagger: vcs$ is the size of the minimum vertex cover of the subgraph induced by the solution; $\star\star: k$ is the number of vertices in the solution; $\alpha(\VV)=\sum_{v\in\VV}\alpha(v)$ where $\alpha(v)$ is the value of the vertex $v$ in \VV; $m$ is the number of edges in the input graph.}	
  \label{tbl:summary}
	\end{center}
\end{table}

We observe that all our problems admit fixed-parameter pseudo-polynomial time algorithms with respect to the treewidth of the graph. Further, all the problems admit an FPTAS for graphs with small (at most $\OO(\log n)$) treewidth. Our results seem to indicate that \shortestpathknapsack is computationally the easiest, followed by \pathknapsack and \pa.

\subsection{Related Work}
To the best of our knowledge, Yamada et al.~\cite{yamada2002heuristic} initiated the study of the \kp problem with a graph structure on the items (which they called conflict graph), where the goal is to compute an independent set which has the maximum total value and satisfies the budget constraint. The paper proposes a set of heuristics and upper bounds based on Lagrangean relaxation. Later, Hifi and Michrafy~\cite{hifi2006reactive,hifi2007reduction} designed a meta-heuristic approach using reactive local search techniques for the problem. Pferschy and Schauer~\cite{PferschyS09} showed that this problem is strongly \NPH and presented FPTAS for some special graph classes. Bettinelli et al.~\cite{BettinelliCM17} proposed a dynamic programming based pruning in a branching based algorithm using the upper bounds proposed by Held et al.~\cite{HeldCS12}. Coniglio et al.~\cite{ConiglioFS21} presented a branch-and-bound type algorithm for the problem. Finally, Luiz et al.~\cite{LuizSU21} proposed a cutting plane based algorithm. Pferschy and Schauer~\cite{PferschyS17}, Gurski and Rehs~\cite{GurskiR19}, and Goebbels et al.~\cite{GoebbelsGK22} showed approximation algorithms for the \kp problem with more sophisticated neighborhood constraints. This problem (and also the variant where the solution should be a clique instead of an independent set) also admits an algorithm running in time $n^{\OO(k)}$ where $k$ is the thinness of the input graph~\cite{DBLP:journals/dam/BonomoE19,DBLP:journals/orl/ManninoORC07}.

Ito et al.~\cite{ito2008approximability} studied FPTASes for a generalized version of our problems, called the {\sc Maximum partition} problem but on a specialized graph class, namely the class of series-parallel graphs. Although their work mentions that their results can be extended to bounded treewidth graphs, the authors state that the algorithmic techniques for that is more complex but do not explain the techniques. To the best of our knowledge, there is no follow-up work where the techniques for bounded treewidth graphs are explained.

Bonomo-Braberman and Gonzalez~\cite{DBLP:journals/dam/Bonomo-Braberman22b} studied a general framework for, which they called "locally checkable problems", which can be used to obtain FPT algorithms for all our problems parameterized by treewidth. However, the running time of these algorithms obtained by their framework is worse than our algorithms parameterized by treewidth.

There are many other generalizations of the basic \kp problem which have been studied extensively in the literature. We refer to \cite{CacchianiILM22,cacchiani2022knapsack,kellerer2004multidimensional,martello1990knapsack} for an overview.

\section{Preliminaries}~\label{sec:prelim}

We denote the set $\{0,1,2,\ldots\}$ of natural numbers with \NB. For any integer \el, we denote the sets $\{1,\ldots,\el\}$ and $\{0,1,\ldots,\el\}$ by $[\el]$ and $[\el]_0$ respectively. Given a graph $G = (V,E)$ the distance between two vertices $u,v \in V$ is denoted by ${\sf dist}_G(u,v)$. We now formally define our problems. \longversion{Our first problem is \pa where the knapsack subset of items must induce a connected subgraph.} We frame all our problems as decision problems so that we can directly use the framework of \NP-completeness.

\begin{definition}[\conknapsack]\label{def:pa}
Given an undirected graph $\GG=(\VV,\EE)$, non-negative weights of vertices $(w(u))_{u\in\VV}$, non-negative values of vertices $(\alpha(u))_{u\in\VV}$, size $s$ of the knapsack, and target value $d$, compute if there exists a subset $\UU\subseteq\VV$ of vertices such that:\shortversion{ (i) \UU is connected, (ii) $w(\UU)=\sum_{u\in\UU} w(u) \le s$, and (iii) $\alpha(\UU)=\sum_{u\in\UU} \alpha(u) \ge d$} 
\longversion{\begin{enumerate}
    \item \UU is connected,
    \item $w(\UU)=\sum_{u\in\UU} w(u) \le s$,
    \item $\alpha(\UU)=\sum_{u\in\UU} \alpha(u) \ge d$.
\end{enumerate}}
We denote an arbitrary instance of \pa by $(\GG,(w(u))_{u\in\VV},(\alpha(u))_{u\in\VV},s,d)$.
\end{definition}
If not mentioned otherwise, for any subset $\VV' \subseteq \VV$, $\alpha(\VV')$ and $w(\VV')$, we denote $\sum_{u \in \VV'} \alpha(u)$ and $\sum_{u \in \VV'} w(u)$, respectively.

In our next problem, the knapsack subset of items must be a path between two given vertices.

\begin{definition}[\pathknapsack]\label{def:path}
	Given an undirected graph $\GG=(\VV,\EE)$, non-negative weights of vertices $(w(u))_{u\in\VV}$, non-negative values of vertices $(\alpha(u))_{u\in\VV}$, size $s$ of the knapsack, target value $d$, two vertices $x$ and $y$, compute if there exists a subset $\UU\subseteq\VV$ of vertices such that:\shortversion{ (i) \UU is a path between $x$ and $y$ in $\GG$, (ii) $w(\UU)=\sum_{u\in\UU} w(u) \le s$, and $\alpha(\UU)=\sum_{u\in\UU} \alpha(u) \ge d$.}
	\longversion{\begin{enumerate}
		\item \UU is a path between $x$ and $y$ in $\GG$,
    \item $w(\UU)=\sum_{u\in\UU} w(u) \le s$,
    \item $\alpha(\UU)=\sum_{u\in\UU} \alpha(u) \ge d$.
	\end{enumerate}}
	We denote an arbitrary instance of \pathknapsack by $(\GG,(w(u))_{u\in\VV},(\alpha(u))_{u\in\VV},s,d,x,y)$.
\end{definition}

In \shortestpathknapsack, the knapsack subset of items must be a shortest path between two given vertices.

\begin{definition}[\shortestpathknapsack]\label{def:path}
Given an undirected edge-weighted graph $\GG=(\VV,\EE,c:\EE\longrightarrow\RB_{\ge0})$, positive weights of vertices $(w(u))_{u\in\VV}$, non-negative values of vertices $(\alpha(u))_{u\in\VV}$, size $s$ of knapsack, target value $d$, two vertices $x$ and $y$, compute if there exists a subset $\UU\subseteq\VV$ of vertices such that:
\begin{enumerate}
	\item \UU is a shortest path between $x$ and $y$ in $\GG$,
    \item $w(\UU)=\sum_{u\in\UU} w(u) \le s$,
    \item $\alpha(\UU)=\sum_{u\in\UU} \alpha(u) \ge d$
\end{enumerate}
We denote an arbitrary instance of \shortestpathknapsack by $(\GG,(w(u))_{u\in\VV},(\alpha(u))_{u\in\VV},s,d,x,y)$.
\end{definition}

If not mentioned otherwise, we use $n,s,d,\GG$ and $\VV$ to denote respectively the number of vertices,  the size of the knapsack, the target value, the input graph, and the set of vertices in the input graph.

%\shortversion{ In the interest of space, we omit the preliminaries on treewidth. It is available in the appendix. We denote the treewidth of a graph \GG by $tw(\GG)$.}

\begin{definition}[Treewidth] 
Let $G = (V_G,E_G)$ be a graph.  A {\em tree-decomposition} of a graph $G$ is a pair 
$(\mathbb{T} = (V_{\mathbb{T}},E_{\mathbb{T}}),\mathcal{ X}=\{X_{t}\}_{t\in V_{\mathbb T}})$,  where 
${\mathbb T}$ is a tree where every node $t\in V_{\mathbb T}$ 
is assigned a subset $X_t\subseteq V_G$, called a bag,  such that the following conditions hold. 
\begin{itemize}
\item $\bigcup_{t\in V_\mathbb{T}}{X_t}=V_G$,
\item for every edge $\{x,y\}\in E_G$ there is a $t\in V_\mathbb{T}$ such that  $x,y\in X_{t}$, and 
\item for any $v\in V_G$ the subgraph of $\mathbb{T}$ induced by the set  $\{t\mid v\in X_{t}\}$ is connected.
\end{itemize}

The {\em width} of a tree decomposition is $\max_{t\in V_\mathbb{T}} |X_t| -1$. The {\em treewidth} of $G$ 
is the  minimum width over all tree decompositions of $G$ and is denoted by $\tw(G)$. 
 
A tree decomposition  $(\mathbb{T},\mathcal{ X})$ is called a {\em nice edge tree decomposition} if $\mathbb{T}$ is a tree rooted at some node $r$ where $ X_{r}=\emptyset$, each node of $\mathbb{T}$ has at most two children, and each node is of one of the following kinds:
\begin{itemize}
\item {\bf Introduce node}: a node $t$ that has only one child $t'$ where $X_{t}\supset X_{t'}$ and  $|X_{t}|=|X_{t'}|+1$.
\item {\bf Introduce edge node} a node $t$ labeled with an edge between
$u$ and $v$, with only one child $t'$ such that $\{u,v\}\subseteq X_{t'}=X_{t}$. 
This bag is said to introduce $uv$. 
\item {\bf  Forget vertex node}: a node $t$ that has only one child $t'$  where $X_{t}\subset X_{t'}$ and  $|X_{t}|=|X_{t'}|-1$.
\item {\bf Join node}:  a node  $t$ with two children $t_{1}$ and $t_{2}$ such that $X_{t}=X_{t_{1}}=X_{t_{2}}$.
\item {\bf Leaf node}: a node $t$ that is a leaf of $\mathbb T$, and $X_{t}=\emptyset$. 
\end{itemize}
We additionally require that every edge is introduced exactly once. 
One can  show that  a tree decomposition of width $t$ can be transformed into 
a nice tree decomposition of the same width $t$ and  with 
 $\mathcal{O}(t |V_G|)$ nodes, see~e.g.~\cite{BODLAENDER201842,DBLP:books/sp/CyganFKLMPPS15}. For a node $t \in \mathbb{T}$, let $\mathbb{T}_t$ be the subtree of $\mathbb{T}$ rooted at $t$, and $V(\mathbb{T}_t)$ denote the vertex set in that subtree. Then $\beta(t)$ is the subgraph of $G$ where the vertex set is  $\bigcup_{t' \in V(\mathbb{T}_t)} X_{t'}$ and the edge set is the union of the set of edges introduced in each $t', t' \in V(\mathbb{T}_t)$. We denote by $V(\beta(t))$ the set of vertices in that subgraph, and by $E(\beta(t))$ the set of edges of the subgraph.
 
%7 authors book

In this paper, we sometimes fix a vertex $v\in V_G$ and include it in every bag of a nice edge tree decomposition $(\mathbb{T},\mathcal{X})$ of $G$, with the effect of the root bag and each leaf bag containing $v$. For the sake of brevity, we also call such a modified tree decomposition a nice tree decomposition. Given the tree $\mathbb{T}$ rooted at the node $r$, for any nodes $t_1,t_2 \in V_\mathbb{T}$, the distance between the two nodes in $\mathbb{T}$ is denoted by $\sf{dist}_\mathbb{T}(t_1,t_2)$.

\end{definition}

\section{\pa}
We present our results for \pa in this section. First, we show that \pa is strongly \NPC by reducing it from \vc, which is known to be \NPC even for $3$-regular graphs~\cite[folklore]{DBLP:journals/dm/FleischnerSS10}. Hence, we do not expect a pseudo-polynomial time algorithm for \pa, unlike \kp.

\begin{definition}[\vc]Given a graph $\GG=(\VV,\EE)$ and a positive integer $k$, compute if there exists a subset $\VV' \subseteq \VV$ such that at least one end-point of every edge belongs to $\VV'$ and $|\VV'|\leq k$. We denote an arbitrary instance of \vc by $(\GG,k)$.
\end{definition}

\begin{theorem}\label{thm:pa-gen-npc}
\pa is strongly \NPC even when the maximum degree of the input graph is four.
\end{theorem}

\begin{proof}
    Clearly, \pa $\in$ \NP. We reduce \vc to \pa to prove NP-hardness. Let $(\GG=(\VV=\{v_i: i\in[n]\},\EE),$ $k)$ be an arbitrary instance of \vc where \GG is $3$-regular. We construct the following instance $(\GG^\pr=(\VV^\pr,\EE^\pr),(w(u))_{u\in\VV^\pr},(\alpha(u))_{u\in\VV^\pr},s,d)$ of \pa.
    \begin{align*}
        &\VV^\pr = \{u_i, g_i: i\in[n]\} \cup \{h_e: e\in \EE\}\\
        &\EE^\pr = \{\{u_i,h_e\}: i\in[n], e \in\EE, e\text{ is incident on }v_i\text{ in }\GG\} \\
        &\cup \{\{u_i,g_i\}: i\in[n]\} \cup \{\{g_i,g_{i+1}\}:i \in [n-1]\}\\
        &w(u_i) = 1, \alpha(u_i)=0, w(g_i)=0, \alpha(g_i)=0, \forall i\in[n]\\
        &w(h_e)=0, \alpha(h_e)=1, \forall e\in\EE, s=k, d=|\EE|
    \end{align*}
    We observe that the maximum degree of $\GG^\pr$ is at most four --- (i) the degree of $u_i$ is four for every $i\in[n]$, since \GG is $3$-regular and $u_i$ has an edge to $g_i$, (ii) the degree of $h_e$ is two for every $e\in\EE$, and (iii) the degree of $g_i$ is at most three for every $i\in[n]$ since the set $\{g_i:i\in[n]\}$ induces a path. We claim that the two instances are equivalent.

    In one direction, let us suppose that the \vc instance is a \yes instance. Let $\WW\subseteq\VV$ be a vertex cover of \GG with $|\WW|\le k$. We consider $\UU=\{u_i: i\in[n],v_i\in\WW\}\cup\{g_i: i\in[n]\}\cup\{h_e: e\in\EE\}\subseteq \VV^\pr$. We claim that $\GG^\pr[\UU]$ is connected. Since $\{g_i:i\in[n]\}$ induces a path and there is an edge between $u_i$ and $g_i$ for every $i\in[n]$, the induced subgraph $\GG^\pr[\{u_i: i\in[n],v_i\in\WW\}\cup\{g_i: i\in[n]\}]$ is connected. Since \WW is a vertex cover of \GG, every edge $e\in\EE$ is incident on at least one vertex in \WW. Hence, every vertex $h_e, e\in\EE,$ has an edge with at least one vertex in $\{u_i: i\in[n],v_i\in\WW\}$ in the graph $\GG^\pr$. Hence, the induced subgraph $\GG^\pr[\UU]$ is connected. Now we have $w(\UU)=\sum_{i=1}^n w(u_i)\mathbbm{1}(u_i\in\UU) + \sum_{i=1}^n w(g_i) + \sum_{e\in\EE} w(h_e)=|\UU|\le k$. We also have $\alpha(\UU)=\sum_{i=1}^n \alpha(u_i)\mathbbm{1}(u_i\in\UU) + \sum_{i=1}^n \alpha(g_i) + \sum_{e\in\EE} \alpha(h_e)=|\EE|$. Hence, the \pa instance is also a \yes instance.

    In the other direction, let us assume that the \pa instance is a \yes instance. Let $\UU\subseteq\VV^\pr$ be a solution of the \pa instance. We consider $\WW=\{v_i: i\in[n]: u_i\in\UU\}$. Since $s=k$, we have $|\WW|\le k$. Also, since $d=|\EE|$, we have $\{h_e: e\in\EE\}\subseteq\UU$. We claim that \WW is a vertex cover of \GG. Suppose not, then there exists an edge $e\in\EE$ which is not covered by \WW. Then none of the neighbors of $h_e$ belongs to \UU contradicting our assumption that \UU is a solution and thus $\GG^\pr[\UU]$ should be connected. Hence, \WW is a vertex cover of \GG and thus the \vc instance is a \yes instance.

    We observe that all the numbers in our reduced \pa instance are at most the number of edges of the graph. Hence, our reduction shows that \pa is strongly \NPC.
\end{proof}

 \Cref{thm:pa-gen-npc} also implies the following corollary in the framework of parameterized complexity.

\begin{corollary}\label{cor:pa-max-deg}
    \pa is \PNPH parameterized by the maximum degree of the input graph.
\end{corollary}

We next show that \pa is \NPC even for trees. For that, we reduce from the \NPC problem \kp.

\begin{definition}[\kp]
Given a set $\XX=[n]$ of $n$ items with sizes $\theta_1,\ldots,\theta_n,$ values $p_1,\ldots,p_n$, capacity $b$ and target value $q$, compute if there exists a subset $\II \subseteq [n]$ such that $\sum_{i\in \II} \theta_i \leq b$ and $\sum_{i\in \II} p_i \geq q$. We denote an arbitrary instance of \kp by $(\XX,(\theta_i)_{i\in\XX}, (p_i)_{i\in\XX}, b,q)$.
\end{definition}

\begin{theorem}\label{thm:pa-star-npc}
	\pa is \NPC even for star graphs.
\end{theorem}

\begin{proof}
\pa clearly belongs to \NP. To show \NP-hardness, we reduce from \kp. Let $(\XX=[n],(\theta_i)_{i\in\XX}, (p_i)_{i\in\XX}, b,q)$ be an arbitrary instance of \kp. We consider the following instance $(\GG(\VV,\EE),(w(u))_{u\in\VV},(\alpha(u))_{u\in\VV},s,d)$ of \pa.
\begin{align*}
    &\VV = \{v_0, v_1, \ldots, v_n\}\\
    &\EE = \{\{v_0,v_i\}: 1\le i\le n\}\\
    &w(v_i) = \theta_i\, \alpha(v_i) = p_i;\forall i\in[n], w(v_0)=\alpha(v_0)=0;\\
    &s = b, d = q
\end{align*}
We now claim that the two instances are equivalent.

In one direction, let us suppose that the \kp instance is a \yes instance. Let $\WW\subseteq\XX$ be a solution of \kp. Let us consider $\UU=\{v_i: i\in\WW\}\cup\{v_0\}\subseteq\VV$. We observe that $\GG[\UU]$ is connected since $v_0\in\UU$. We also have
\[ w(\UU) = \sum_{i\in\WW} w(v_i) = \sum_{i\in\WW} \theta_i \le b = s,\]
and
\[ \alpha(\UU) = \sum_{i\in\WW} \alpha(v_i) = \sum_{i\in\WW} p_i \ge q = d.\]
Hence, the \pa instance is a \yes instance.

In the other direction, let us assume that the \pa instance is a \yes instance with $\UU\subseteq\VV$ be one of its solution. Let us consider a set $\WW=\{i: i\in[n], v_i\in\UU\}$. We now have
\[ \sum_{i\in\WW} \theta_i = \sum_{i\in\WW} w(v_i) = w(\UU) \le s = b, \]
and
\[ \sum_{i\in\WW} p_i = \sum_{i\in\WW} \alpha(v_i) = \alpha(\UU) \ge d=q.\]
Hence, the \kp instance is a \yes instance.
\end{proof}

We complement the hardness result in \Cref{thm:pa-star-npc} by designing a pseudo-polynomial-time algorithm for \pa for trees. In fact, we have designed an algorithm with running time $2^{\OO(\tw\log \tw)}\cdot n^{\mathcal{O}(1)}\cdot {\sf min}\{s^2,d^2\}$ where the treewidth of the input graph is $\tw$. We present this algorithm next.

\subsection{Treewidth as a Parameter}
% \subsection{\pa on Bounded Treewidth Graphs}
%\subsection{Treewidth as Parameter}
% In this part, we study \pa parameterized by the treewidth of the input graph $G$. Note that since the problem is \NPC even for stars, we do not expect FPT algorithms parameterized by the treewidth of $G$. However, we design an algorithm with running time $f(k)\cdot p(n,s,d)$; here $f$ is a computable function on $k = tw(G)$, and $p$ is a polynomial dependent on $n = |V(G)|$, the target size $s$ and the target value $d$.

\begin{theorem}\label{thm:treewidth-pa}
There is an algorithm for \pa with running time $2^{\OO(\tw\log \tw)}\cdot n\cdot {\sf min}\{s^2,d^2\}$ where $n$ is the number of vertices in the input graph, $\tw$ is the treewidth of the input graph, $s$ is the input size of the knapsack and $d$ is the input target value.
\end{theorem}

\begin{proof}
Let $(G = (V_G,E_G),{(w(u))_{u \in V_G}, (\alpha(u))_{u\in V_G}}, s,d)$ be an input instance of \pa such that $\tw=tw(G)$. Let $\mathcal{U} \subseteq V_G$ be a solution subset for the input instance. For technical purposes, we guess a vertex $v \in \mathcal{U}$ --- once the guess is fixed, we are only interested in finding solution subsets $\mathcal{U}'$ that contain $v$ and $\mathcal{U}$ is one such candidate. We also consider a nice edge tree decomposition $(\mathbb{T} = (V_{\mathbb{T}},E_{\mathbb{T}}),\mathcal{X})$ of $G$ that is rooted at a node $r$, and where $v$ has been added to all bags of the decomposition. Therefore, $X_{r} = \{v\}$ and each leaf bag is the singleton set $\{v\}$.

We define a function $\ell: V_{\mathbb{T}} \rightarrow \mathbb{N}$.
%For the root $r$, $\ell(r) = 0$. 
For a vertex $t \in V_\mathbb{T}$, $\ell(t) = \sf{dist}_\mathbb{T}(t,r)$, where $r$ is the root. Note that this implies that $\ell(r) = 0$. Let us assume that the values that $\ell$ takes over the nodes of $\mathbb{T}$ is between $0$ and $L$. Now, we describe a dynamic programming algorithm over $(\mathbb{T},\mathcal{X})$ for \pa.

%%%%%%%%
    \paragraph{\textbf{States.}} We maintain a DP table $D$ where a state has the following components:
    \begin{enumerate}
    \item $t$ represents a node in $V_\mathbb{T}$.
    \item $\mathbb{P} = (P_0,\ldots,P_m), m\leq \tw+1$ represents a partition of the vertex subset $X_t$. 
    \end{enumerate}
%%%%%%
    \paragraph{\textbf{Interpretation of States.}} For a state $[t,\mathbb{P}]$, if there is a solution subset $\mathcal{U}$ let $\mathcal{U'} = \mathcal{U} \cap V(\beta(t))$. Let $\beta_{\mathcal{U'}}$ be the graph induced on $\mathcal{U'}$ in $\beta(t)$. Let $C_1,C_2,\ldots,C_m$ be the connected components of $\beta_{\mathcal{U'}}$. Note that $m\leq \tw+1$. Then in the partition $\mathbb{P} = (P_0,P_1,\ldots,P_m)$, $P_i = C_i \cap X_t, 1 \leq i \leq m$. Also, $P_0 = X_t \setminus \mathcal{U'}$. 

Given a node $t\in V_{\mathbb{T}}$, a subgraph $H$ of $\beta(t)$ is said to be a $\mathbb{P}$-subgraph if (i) the connected components $C_1,C_2,\ldots,C_m$ of $H$, $m\leq \tw+1$ are such that $P_i = C_i \cap X_t, 1 \leq i \leq m$, (ii) $P_0 = X_t \setminus H$. For each state $[t,\mathbb{P}]$, a pair $(w,\alpha)$ with $w\leq s$ is said to be feasible if there is a $\mathbb{P}$-subgraph of $\beta(t)$ whose total weight is $w$ and total value is $\alpha$. Moreover, a feasible pair $(w,\alpha)$ is said to be undominated if there is no other $\mathbb{P}$-subgraph with weight $w'$ and value $\alpha'$ such that $w' \leq w$ and $\alpha' \geq \alpha$. Please note that by default, an empty $\mathbb{P}$-subgraph has total weight $0$ and total value $0$.

For each state $[t,\mathbb{P}]$, we initialize $D[t,\mathbb{P}]$ to the list $\{(0,0)\}$. Our computation shall be such that in the end each $D[t,\mathbb{P}]$ stores the set of all undominated feasible pairs $(w,\alpha)$ for the state $[t,\mathbb{P}]$.

%%%%%%
    \paragraph{\textbf{Dynamic Programming on $D$.}} We describe the following procedure to update the table $D$. We start updating the table for states with nodes $t\in V_\mathbb{T}$ such that $\ell(t)=L$. When all such states are updated, then we move to update states where the node $t$ has $\ell(t) = L-1$, and so on till we finally update states with $r$ as the node --- note that $\ell(r) =0$. For a particular $i$, $0\leq i< L$ and a state $[t,\mathbb{P}]$ such that $\ell(t) = i$, we can assume that $D[t',\mathbb{P}']$ have been evaluated for all $t'$ such that $\ell(t')>i$ and all partitions $\mathbb{P}'$ of $X_{t'}$. Now we consider several cases by which $D[t,\mathbb{P}]$ is updated based on the nature of $t$ in $\mathbb{T}$:
    \begin{enumerate}
    \item Suppose $t$ is a leaf node. Note that by our modification, $X_t = \{v\}$. There can be only 2 partitions for this singleton set --- $\mathbb{P}_t^1 = (P_0 = \emptyset, P_1 = \{v\})$ and $\mathbb{P}_t^2 = (P_0 = \{v\}, P_1 = \emptyset)$. If $\mathbb{P} = \mathbb{P}_t^1$ then $D[t,\mathbb{P}]$ stores the pair $(w(v),\alpha(v))$ if $w(v) \leq s$ and otherwise no modification is made. If $\mathbb{P} = \mathbb{P}_t^2$ then $D[t,\mathbb{P}]$ is not modified.
    
    \item Suppose $t$ is a forget vertex node. Then it has an only child $t'$ where $X_t \subset X_{t'}$ and there is exactly one vertex $u \neq v$ that belongs to $X_{t'}$ but not to $X_t$. Let $\mathbb{P}' = (P'_0,P'_1,\ldots,P'_{m'})$ be a partition of $X_{t'}$ such that when restricted to $X_t$ we obtain the partition $\mathbb{P} = (P_0,P_1,\ldots,P_m )$. For each such partition, we shall do the following.\\ Suppose $\mathbb{P}'$ has $u \in P'_0$, then each feasible undominated pair stored in $D[t',\mathbb{P}']$ is copied to $D[t,\mathbb{P}]$. \\
    Alternatively, suppose $\mathbb{P}'$ has $u \in P'_i, i>0$ and $|P'_i| >1$. Then, each feasible undominated pair stored in $D[t',\mathbb{P}']$ is copied to $D[t,\mathbb{P}]$.\\
    Finally, suppose $\mathbb{P}'$ has $u \in P'_i, i>0$ and $P'_i = \{u\}$. Then we do not make any changes to $D[t,\mathbb{P}]$.
    
    \item Suppose $t$ is an introduce node. Then it has an only child $t'$ where $X_{t'} \subset X_{t}$ and there is exactly one vertex $u \neq v$ that belongs to $X_{t}$ but not $X_{t'}$. Note that no edges incident to $u$ have been introduced yet, and so in $\beta(t)$ $u$ is not yet connected to any other vertex. Let $\mathbb{P}' = (P'_0,P'_1,\ldots,P'_{m'})$ be the partition of $X_{t'}$ obtained from restricting the partition $\mathbb{P} = (P_0,P_1,\ldots,P_m )$ to $X_{t'}$. First, suppose $u \in P_0$. Then we copy all pairs of $D[t',\mathbb{P}']$ to $D[t,\mathbb{P}]$. \\
    Next, suppose $u \in P_i, i>0$ and $P_i = \{u\}$. Then for each pair $(w,\alpha)$ in $D[t',\mathbb{P}']$, if $w + w(u) \leq s$ we add $(w+w(u),\alpha+\alpha(u))$ to the set in $D[t,\mathbb{P}]$.\\
    Finally, let $u \in P_i, i>0$ and $|P_i| >1 $. Then we make no changes to $D[t,\mathbb{P}]$.
    
    \item Suppose $t$ is an introduce edge node. Then it has an only child $t'$ where $X_{t'} = X_{t}$, and additionally for two vertices $u,w \in X_t = X_{t'}$, the edge $\{u,w\}$ is introduced into $\beta(t)$. First, suppose $\mathbb{P} = (P_0,P_1,\ldots,P_m)$ is such that one of $u,w$ is in $P_0$. Then all pairs of $D[t',\mathbb{P}]$ are copied to $D[t,\mathbb{P}]$.\\
    Next, let $u \in P_i$, $w\in P_j$, $i\neq j \neq 0$. Then no updates are made to $D[t,\mathbb{P}]$. \\
    Finally, suppose $u,w \in P_i, i>0$. Copy to $D[t,\mathbb{P}]$ all pairs from $D[t',\mathbb{P}]$. Consider a partition $\mathbb{P}'$ where the part $P'$ contains $u$ and $P''$ contains $w$. $\mathbb{P}'$ is such that $P_i = P' \cup P''$ and any other $P_j$ with $j \neq i$ is a part in $\mathbb{P}'$. Any pair of $D[t',\mathbb{P}']$ is copied to $D[t,\mathbb{P}]$.
    \item Suppose $t$ is a join node. Then it has two children $t_1,t_2$ such that $X_t = X_{t_1} = X_{t_2}$. Consider $\mathbb{P} = (P_0,P_1,\ldots,P_m)$. Let $(w_{\mathbb{P}}, \alpha_{\mathbb{P}})$ be the total weight and value of the vertices in $\cup_{1\leq i \leq m} P_i$. Consider a pair $(w_1,\alpha_1)$ in $D[t_1,\mathbb{P}]$ and a pair $(w_2,\alpha_2)$ in $D[t_2,\mathbb{P}]$. Suppose $w_1 + w_2 - w_{\mathbb{P}} \leq s$. Then we add $(w_1 + w_2 - w_{\mathbb{P}}, \alpha_1+\alpha_2-\alpha_{\mathbb{P}})$ to $D[t,\mathbb{P}]$.
    \end{enumerate}
    
Finally, in the last step of updating $D[t,\mathbb{P}]$, we go through the list saved in $D[t,\mathbb{P}]$ and only keep undominated pairs.

The output of the algorithm is a pair $(w,\alpha)$ stored in $D[r,\mathbb{P} = (P_0 = \emptyset, P_1 = \{v\})]$ such that $w \leq s$ and $\alpha$ is the maximum value over all pairs in $D[r,\mathbb{P}]$.

%%%%%%%
    \paragraph{\textbf{Correctness of the Algorithm.}}

    Recall that we are looking for a solution $\mathcal{U}$ that contains the fixed vertex $v$ that belongs to all bags of the tree decomposition. First, we show that a pair $(w,\alpha)$ belonging to $D[t,\mathbb{P}]$ for a node $t \in V_\mathbb{T}$ and a partition $\mathbb{P}$ of $X_t$ corresponds to a $\mathbb{P}$-subgraph $H$ in $\beta(t)$. Recall that $X_r = \{v\}$. Thus, this implies that a pair $(w,\alpha)$ belonging to $D[r,\mathbb{P} = (P_0 = \emptyset, P_1 = \{v\})]$ corresponds to a connected subgraph of $G$. Moreover, the output is a pair that is feasible and with the highest value. 

    In order to show that a pair $(w,\alpha)$ belonging to $D[t,\mathbb{P}]$ for a node $t \in V_\mathbb{T}$ and a partition $\mathbb{P}$ of $X_t$ corresponds to a $\mathbb{P}$-subgraph $H$ in $\beta(t)$, we need to consider the cases of what $t$ can be:
    \begin{enumerate}
    \item When $t$ is a leaf node with $\ell(t) = i$, $X_t$ only contains $v$ and the update to $D$ is done such that $v$ is the corresponding subgraph to a stored pair. This is true in particular when $i =L$, the base case. From now we can assume that for a node $t$ with $\ell(t) = i < L$ all $D[t',\mathbb{P}']$ entries are correct and correspond to $\mathbb{P}'$-subgraphs in $\beta(t')$ when $\ell(t') > i$.
    \item When $t$ is a forget vertex node, let $t'$ be the child node and $u \neq v$ be the vertex that is being forgotten. We copy pairs from $D[t',\mathbb{P}']$ depending on the structure of $\mathbb{P}'$. Since $\ell(t') > \ell(t)$, by induction hypothesis all entries in $D[t',\mathbb{P}']$ for any partition $\mathbb{P}'$ of $X_{t'}$ are feasible. From the cases considered, we copy a pair to $D[t,\mathbb{P}]$ from a $D[t',\mathbb{P}']$ only when $u$ is not part of the $\mathbb{P}'$-subgraph or is in a component of $\mathbb{P}'$ that has vertices in $X_t$. Thus, the same subgraph is a $\mathbb{P}$-subgraph in $\beta(t)$.
    \item When $t$ is an introduce node, there is a child $t'$ we are introducing a vertex $u \neq v$ that has no adjacent edges added in $\beta(t)$. Since $\ell(t') > \ell(t)$, by induction hypothesis all entries in $D[t',\mathbb{P}']$ for any partition $\mathbb{P}'$ of $X_{t'}$ are feasible. We update pairs in $D[t,\mathbb{P}]$ from $D[t',\mathbb{P}']$ such that either $u$ is not considered as part of a $\mathbb{P}$-subgraph and the pair is certified by the $\mathbb{P}'$-subgraph, or $u$ is added to a $\mathbb{P}'$-subgraph in order to obtain a new $\mathbb{P}$-subgraph.
    \item When $t$ is an introduce edge node, there is a child $t'$ such that $X_t = X_{t'}$ and the only difference is that two vertices $u,w$ in the bags $X_t = X_{t'}$ now have an edge in $\beta(t)$. Since $\ell(t') > \ell(t)$, by induction hypothesis all entries in $D[t',\mathbb{P}']$ for any partition $\mathbb{P}'$ of $X_{t'}$ are feasible. The updates are made in the cases when one of $u$ or $w$ is not in the intended $\mathbb{P}$-subgraph and the included pair is certified by a $\mathbb{P'}$-subgraph, or when the $u$ and $w$ belong to different components of a $\mathbb{P}'$-subgraph and the new $\mathbb{P}$-subgraph has these components merged as a single component.
    \item When $t$ is a join node, there are two children $t_1,t_2$ such that $X_T = X_{t_1} = X_{t_2}$. Note that this implies that $V(\beta(t_1)) \cap V(\beta(t_2)) = X_t$. Since $\ell(t_1),\ell(t_2) > \ell(t)$, by induction hypothesis all entries in $D[t_i,\mathbb{P}']$ for any partition $\mathbb{P}'$ of $X_{t_i}$ are feasible for $i \in \{1,2\}$.We update pairs in $D[t,\mathbb{P}]$ when there is a $\mathbb{P}$-subgraph in $\beta(t_1)$ and a $\mathbb{P}$-subgraph in $\beta(t_2)$ and we take the union of these two subgraphs to obtain a $\mathbb{P}$-subgraph in $\beta(t)$.
    \end{enumerate}
    Thus in all cases of $t$, a pair added to $D[t,\mathbb{P}]$ for some partition of $X_t$ is a feasible pair. Recall that as a last step of updating $D[t,\mathbb{P}]$, we go through the entire list and keep only  undominated pairs in the list.

    What remains to be shown is that an undominated feasible solution $\mathcal{U}$ of \pa in $G$ is contained in $D[r,\mathbb{P} = (P_0 = \emptyset, P_1 = \{v\})]$. Let $w$ be the weight of $\mathcal{U}$ and $\alpha$ be the value. Recall that $v \in \mathcal{U}$. For each $t$, we consider the subgraph $\beta(t) \cap \mathcal{U}$. Let $C_1,C_2,\ldots,C_m$ be components of $\beta(t) \cap \mathcal{U}$ and let for each $1\leq i \leq m, P_i = X_t \cap C_i$. Also, let $P_0 = X_t \setminus \mathcal{U}$. Consider $\mathbb{P} = (P_0,P_1,\ldots,P_m)$. The algorithm updates in $D[t,\mathbb{P}]$ the pair $(w',\alpha')$ for the subsolution $\beta(t) \cap \mathcal{U}$. Therefore, $D[r,\mathbb{P} = (\emptyset,\{v\})]$ contains the pair $(w,\alpha)$. Thus, we are done.
%%%%%%
    \paragraph{\textbf{Running time.}} There are $n$ choices for the fixed vertex $v$. Upon fixing $v$ and adding it to each bag of $(\mathbb{T}, \mathcal{X})$ we consider the total possible number of states.  There are at most $\OO(n)\cdot 2^{\tw\log \tw}$ states. For each state, since we are keeping only undominated pairs, for each $w$ there can be at most one pair with $w$ as the first coordinate; similarly, for each $\alpha$ there can be at most one pair with $\alpha$ as the second coordinate. Thus, the number of undominated pairs in each $D[t,\mathbb{P}]$ is at most ${\sf min}\{s,d\}$. By the description of the algorithm, the maximum length of the list stored at $D[t,\mathbb{P}]$ during updation, but before the check is made for only undominated pairs, is $2^{\tw \log \tw}\cdot{\sf min}\{s^2,d^2\}$. Thus, updating the DP table at any vertex takes $2^{\tw \log \tw}\cdot{\sf min}\{s^2,d^2\}$ time. Since there are $\OO(n\cdot\tw)$ vertices in $\mathbb{T}$, the total running time of the algorithm is $2^{\OO(\tw\log \tw)}\cdot n\cdot {\sf min}\{s^2,d^2\}$.
\end{proof}

\subsection{A Fixed Parameter Fully Pseudo-polynomial Time Approximation Scheme}

We now use the algorithm in \Cref{thm:treewidth-pa} as a black-box to design an $(1-\eps)$-factor approximation algorithm for optimizing the value of the solution and running in time $2^{\tw\log \tw}\cdot \text{poly}(n,1/\eps)$.

\begin{theorem}\label{thm-fptas}
There is an $(1-\eps)$-factor approximation algorithm for \pa for optimizing the value of the solution and running in time $2^{\OO(\tw\log \tw)}\cdot \text{poly}(n,1/\eps)$ where \tw is the treewidth of the input graph.
\end{theorem}

\begin{proof}
Let $\II=(\GG=(\VV,\EE),(w(u))_{u\in\VV},(\alpha(u))_{u\in\VV},s)$ be an arbitrary instance of \pa where the goal is to output a connected subgraph \UU of maximum $\alpha(\UU)$ subject to the constraint that $w(\UU)\le s$. Without loss of generality, we can assume that $w(u)\le s$ for every $u\in\VV$. If not, then we can remove every $u\in\VV$ whose $w(u)>s$; this does not affect any solution since any vertex deleted can never be part of any solution. Let $\alpha_{\text{max}}=\max\{\alpha(u): u\in\VV\}$. We construct another instance $\II^\pr=\left(\GG=(\VV,\EE),(w(u))_{u\in\VV},(\alpha^\pr(u)=\left\lfloor \frac{n\alpha(u)}{\eps\alpha_{\text{max}}}\right\rfloor)_{u\in\VV},s\right)$ of \pa. We compute the optimal solution $\WW^\pr\subseteq\VV$ of $\II^\pr$ using the algorithm in \Cref{thm:treewidth-pa} and output $\WW^\pr$. Let $\WW\subseteq\VV$ be an optimal solution of $\II$. Clearly $\WW^\pr$ is a valid (may not be optimal) solution of \II also, since $w(\WW^\pr)\le s$ by the correctness of the algorithm in \Cref{thm:treewidth-pa}. We now prove the approximation factor of our algorithm.
\begin{align*}
     \sum_{u\in\WW^\pr}\alpha(u) &\ge \frac{\eps\alpha_{\text{max}}}{n}\sum_{u\in\WW^\pr} \left\lfloor \frac{n\alpha(u)}{\eps\alpha_{\text{max}}}\right\rfloor\\
    &\ge \frac{\eps\alpha_{\text{max}}}{n}\sum_{u\in\WW} \left\lfloor \frac{n\alpha(u)}{\eps\alpha_{\text{max}}}\right\rfloor&\text{[since $\WW^\pr$ is an optimal solution of $\II^\pr$]}\\
    &\ge \frac{\eps\alpha_{\text{max}}}{n}\sum_{u\in\WW} \left( \frac{n\alpha(u)}{\eps\alpha_{\text{max}}}-1\right)\\
    &\ge \left(\sum_{u\in\WW}\alpha(u)\right)-\eps\alpha_{\text{max}}\\
    &\ge \OPT(\II)-\eps\OPT(\II) & \text{[$\alpha_{\text{max}}\le\OPT(\II)$]}\\
    &=(1-\eps)\OPT(\II)
\end{align*}
Hence, the approximation factor of our algorithm is $(1-\eps)$. We now analyze the running time of our algorithm.

The value of any optimal solution of $\II^\pr$ is at most
\[\sum_{u\in\VV}\alpha^\pr(u) \le \frac{n}{\eps\alpha_{\text{max}}}\sum_{u\in\VV}\alpha(u)\le \frac{n}{\eps\alpha_{\text{max}}}\sum_{u\in\VV}\alpha_{\text{max}}=\frac{n^2}{\eps}. \]
Hence, the running time of our algorithm is $2^{\tw\log \tw}\cdot \text{poly}(n,1/\eps)$.
\end{proof}

\subsubsection{Other Parameters}

We next consider $vcs$, the maximum size of a minimum vertex cover of the subgraph induced by any solution of \pa, as our parameter. That is, $vcs(\II=(\GG,(w(u))_{u\in\VV},(\alpha(u))_{u\in\VV},s,d)) = \max\{\text{size of minimum vertex cover of }$ $W: W\subseteq\VV\text{ is a solution of }\II\}.$ We already know from \Cref{thm:pa-star-npc} that \pa is \NPC for star graphs. We note that $vcs$ is one for star graphs. Hence, \pa is \PNPH with respect to $vcs$, that is, there is no algorithm for \pa which runs in polynomial time even for constant values of $vsc$. However, whether there exists any algorithm with running time $\OO(f(vcs).\text{poly}(n,s,d))$, remains a valid question. We answer this question negatively in \Cref{thm:vcs-woh}. For that, we exhibit an \FPT-reduction from \pvc which is known to be \WOH parameterized by the size of the partial vertex cover we are looking for.

\begin{definition}[\pvc]
Given a graph $\GG=(\VV,\EE)$ and two integers $k$ and $\el$, compute if there exists a subset $\VV^\pr\subseteq\VV$ such that (i) $|\VV^\pr|\le k$ and (ii) there exist at least $\el$ edges whose one or both end points belong to $\VV^\pr$. We denote an arbitrary instance of \pvc by $(\GG,k,\el)$.
\end{definition}

% We know that \pvc is \WOH parameterized by $k$. 

\begin{theorem}\label{thm:vcs-woh}
There is no algorithm for \pa running in time $\OO(f(vcs).\text{poly}(n,s,d))$ unless \ETH fails.
\end{theorem}

\begin{proof}
At a high-level, our reduction from \pvc is similar to the reduction in \Cref{thm:pa-gen-npc}. The only difference is that we will now have only one ``global vertex'' instead of the ``global path'' in the proof of \Cref{thm:pa-gen-npc}. Formally, our reduction is as follows.

Let $(\GG=(\VV,\EE),k,\el)$ be an arbitrary instance of \pvc. We consider the following instance $(\GG^\pr=(\VV^\pr,\EE^\pr),(w(u))_{u\in\VV^\pr},(\alpha(u))_{u\in\VV^\pr},s,d)$ of \pa.
\begin{align*}
    &\VV^\pr = \{u_i: i\in[n]\} \cup \{h_e: e\in \EE\}\cup\{g\}\\
    &\EE^\pr = \{\{u_i,h_e\}: i\in[n], e \in\EE, e\text{ is incident on }v_i\text{ in }\GG\}
    \cup \{\{u_i,g\}: i\in[n]\}\\
    &w(u_i) = 1, \alpha(u_i)=0, \forall i\in[n], 
    w(h_e)=0, \alpha(h_e)=1, \forall e\in\EE,\\
    &w(g)=0, \alpha(g)=0, s=k, d=\el
\end{align*}
We claim that the two instances are equivalent.

In one direction, let us suppose that the \pvc instance is a \yes instance. Let $\WW\subseteq\VV$ covers at least \el edges in \GG and $|\WW|\le k$. We consider $\UU=\{u_i: i\in[n],v_i\in\WW\}\cup\{h_e: e\in\EE, e\text{ is covered by }\WW\}\cup\{g\}\subseteq \VV^\pr$. We claim that $\GG^\pr[\UU]$ is connected. Since there is an edge between $u_i$ and $g$ for every $i\in[n]$, the induced subgraph $\GG^\pr[\{u_i: i\in[n],v_i\in\WW\}\cup\{g\}]$ is connected. Also, since $h_e$ belongs to \UU only if \WW covers $e$ in \GG, the induced subgraph $\GG^\pr[\UU]$ is connected. Now we have $w(\UU)=\sum_{i=1}^n w(u_i)\mathbbm{1}(u_i\in\UU) + w(g) + \sum_{e\in\EE, \WW\text{ covers }e} w(h_e)=|\WW|\le k$. We also have $\alpha(\UU)=\sum_{i=1}^n \alpha(u_i)\mathbbm{1}(u_i\in\UU) + \alpha(g) + \sum_{e\in\EE, \WW\text{ covers }e} \alpha(h_e)\ge\el$. Hence, the \pa instance is also a \yes instance.

In the other direction, let us assume that the \pa instance is a \yes instance. Let $\UU\subseteq\VV^\pr$ be a solution of the \pa instance. We consider $\WW=\{v_i: i\in[n]: u_i\in\UU\}$. Since $s=k$, we have $|\WW|\le k$. Also, since $d=\el$, we have $|\{h_e: e\in\EE\}\cap\UU|\ge\el$. We claim that \WW covers every edge in $\{e: e\in\EE, h_e\in\UU\}$. Suppose not, then there exists an edge $e\in\{e: e\in\EE, h_e\in\UU\}$ which is not covered by \WW. Then none of the neighbors of $h_e$ belongs to \UU contradicting our assumption that \UU is a solution and thus $\GG^\pr[\UU]$ should be connected. Hence, \WW covers at least \el edges in \GG and thus the \pvc instance is a \yes instance.

We also observe that the size of a minimum vertex cover of the subgraph induced by any solution of the reduced \pa instance is at most $k+1$ --- the set consisting of the vertices in any solution of the \pa instance in $\GG^\pr$ which has weight $1$ and $g$ forms a vertex cover of the subgraph induced by the solution of the \pa instance. Also, all the numbers in our reduced \pa instance are at most the number of edges of the graph, and \pvc is \WOH parameterized by $k$. Hence, \pa is \WOH parameterized by the maximum size of the minimum vertex cover of any solution even when all the numbers are encoded in unary. Therefore, there is no algorithm for \pa running in time $\OO(f(vcs).\text{poly}(n,s,d))$ unless \ETH fails.
\end{proof}
\section{\pathknapsack}

We now present the results of \pathknapsack. We first show that \pathknapsack is strongly \NPC by reducing from \hp which is defined as follows.

\begin{definition}[\hp]
Given a graph $\GG(\VV,\EE)$ and two vertices $x$ and $y$, compute if there exists a path between $x$ and $y$ which visits every other vertex in \GG. We denote an arbitrary instance of \hp by $(\GG,x,y)$.
\end{definition}

\hp is known to be \NPC even for graphs with maximum degree three~\cite{DBLP:conf/stoc/GareyJS74}.

\begin{theorem}\label{thm:path-gen-npc}
\pathknapsack is strongly \NPC even for graphs with maximum degree three.
\end{theorem}

\begin{proof}
\pathknapsack clearly belongs to \NP. To show \NP-hardness, we reduce from \hp. Let $(\GG,x,y)$ be an arbitrary instance of \hp. We consider the following instance $(\GG,(w(u))_{u\in\VV},(\alpha(u))_{u\in\VV},s,d,x,y)$ of \pathknapsack.
\[w(u)=0,\alpha(u)=1~\forall u\in\VV,s=0,d=n\]
We claim that the two instances are equivalent.

In one direction, let us assume that \hp is \yes. Let \PP be a Hamiltonian path between $x$ and $y$ in \GG. Then we have $w(\PP)=0\le s$ and $\alpha(\PP)=n\ge d$. Hence, the \pathknapsack instance is also a \yes instance.

In the other direction, let us assume that the \pathknapsack instance is a \yes instance. Let \PP be a path between $x$ and $y$ with $w(\PP)\le s=0$ and $\alpha(\PP)\ge d=n$. Since $\alpha(\VV)=n=\alpha(\PP)$ and there is no vertex with zero $\alpha$ value, \PP must be a Hamiltonian path between $x$ and $y$. Hence, the \hp instance is also a \yes instance.

We observe that all the numbers in our reduced \pathknapsack instance are at most the number of vertices in the graph. Hence, our reduction shows that \pathknapsack is strongly \NPC.
\end{proof}

However, \pathknapsack is clearly polynomial-time solvable for trees, since there exists only one path between every two vertices in any tree.

\begin{observation}\label{obs:pk-tree-poly}
\pathknapsack is polynomial-time solvable for trees.
\end{observation}

One immediate natural question is if \Cref{obs:pk-tree-poly} can be generalized to graphs of bounded treewidth. The following result refutes the existence of any such algorithm.

\begin{theorem}\label{thm:pk-pathwidth}
\pathknapsack is \NPC even for graphs of pathwidth at most two. In particular, \pathknapsack is \PNPH parameterized by pathwidth.
\end{theorem}

\begin{proof}
We show a reduction from \kp. Let $(\XX=[n],(\theta_i)_{i\in\XX}, (p_i)_{i\in\XX}, b,q)$ be an arbitrary instance of \kp. We consider the following instance $(\GG=(\VV,\EE),(w(u))_{u\in\VV},(\alpha(u))_{u\in\VV},s,d,x,y)$ of \pathknapsack.
\begin{align*}
\VV &= \{u_0\}\cup\{u_i,v_i,w_i: i\in[n]\}\\
\EE &= \{\{u_i,v_{i+1}\},\{u_i,w_{i+1}\}:0\le i\le n-1\}\cup\{\{u_i,v_{i}\},\{u_i,w_{i}\}:i\in[n]\}\\
w(v_i)&=\theta_i, \alpha(v_i)=p_i ~\forall i\in[n], w(u_i)=\alpha(u_i)=0~\forall i\in[n]_0, w(w_i)=\alpha(w_i)=0~\forall i\in[n]\\
s &= b, d=q, x=u_0, y=u_n
\end{align*}
\Cref{fig:enter-label} shows a schematic diagram of the reduced instance. We now claim that the two instances are equivalent.

\begin{figure}
    \centering
    \includegraphics[width=\linewidth]{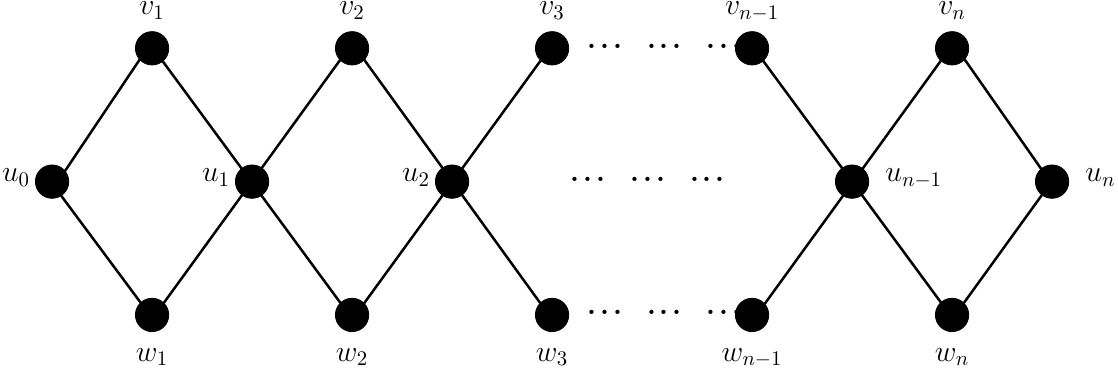}
    \caption{Reduced instance of \pathknapsack in \Cref{thm:pk-pathwidth}.}
    \label{fig:enter-label}
\end{figure}

In one direction, let us suppose that the \kp instance is a \yes instance. Let $\II\subseteq\XX$ be such that $\sum_{i\in\II}\theta_i\le b$ and $\sum_{i\in\II}p_i\ge q$. We now consider the path \PP from $u_0$ to $u_n$ consisting of $2n$ edges where the $2i$-th vertex (from $u_0$) is $v_i$ if $i\in\II$ and $w_i$ otherwise for every $i\in[n]$; the $(2i-1)$-th vertex (from $u_0$) is $u_{i-1}$ for every $i\in[n]$ and the last vertex is $u_n$. Clearly, we have $w(\PP)=\sum_{i\in\II}\theta_i\le b=s$ and $\alpha(\PP)=\sum_{i\in\II}p_i\ge q=d$. Hence, the \pathknapsack instance is a \yes instance.

In the other direction, let us assume that the \pathknapsack instance is a \yes instance. Let \PP be a path between $x$ and $y$ with $w(\PP)\le s=b$ and $\alpha(\PP)\ge d=q$. We observe that all paths from $u_0$ to $u_n$ in \GG have $2n$ edges; $(2i-1)$-th vertex is $u_{i-1}$ and $2i$-th vertex is either $v_i$ or $w_i$ for every $i\in[n]$ and the last vertex is $u_n$. We consider the set $\II=\{i\in[n]: \text{the $(2i-1)$-th vertex in $\PP$ is $v_i$}\}\subseteq\XX$. We now have $\sum_{i\in\II}\theta_i=w(\PP)\le s=b$ and $\sum_{i\in\II}p_i=\alpha(\PP)\ge d=q$. Hence, the \kp instance is a \yes instance.

We also observe that the pathwidth of \GG is at most two (there is a path decomposition with bag size being at most three) which proves the result.
\end{proof}

\Cref{thm:pk-pathwidth} leaves the following question open: does there exist an algorithm for \pathknapsack which runs in time $\OO(f(\tw)\cdot \text{poly}(n,s,d))$? We answer this question affirmatively in the following result. The dynamic programming algorithm is very similar to \Cref{thm:treewidth-pa}. 

\begin{theorem}\label{thm:treewidth-path}
There is an algorithm for \pathknapsack with running time $2^{\OO(\tw\log \tw)}\cdot n\cdot {\sf min}\{s^2,d^2\}$ where $\tw$ is the treewidth of the input graph, $s$ is the input size of the knapsack and $d$ is the input target value.
\end{theorem}

Using the technique in \Cref{thm-fptas}, we use \Cref{thm:treewidth-path} as a black-box to obtain the following approximation algorithm. We again omit its proof due to its similarity with \Cref{thm-fptas}.

\begin{theorem}\label{thm-fptas-path}
There is an $(1-\eps)$ factor approximation algorithm for \pathknapsack for optimizing the value of the solution running in time $2^{\OO(\tw\log \tw)}\cdot \text{poly}(n,1/\eps)$ where \tw is the treewidth of the input graph.
\end{theorem}

We next consider the size of the minimum vertex cover $vcs$ of the subgraph induced by a solution $\WW\subseteq\VV[\GG]$. We observe that the size of the minimum vertex cover of $\GG[\WW]$ is at least half of $|\WW|$ since there exists a Hamiltonian path in $\GG[\WW]$. Hence, it is enough to design an \FPT algorithm with parameter $|\WW|$. Our algorithm is based on the color coding technique~\cite{DBLP:books/sp/CyganFKLMPPS15}.

\begin{theorem}\label{thm:sol-path}
    There is an algorithm for \pathknapsack running in time $\OO\left((2e)^k k^{\OO(\log k)}n^{\OO(1)}\right)$ where $k$ is the number of vertices in the solution.
\end{theorem}

\begin{proof}
Let $(\GG=(\VV,\EE),(w(u))_{u\in\VV},(\alpha(u))_{u\in\VV},s,d,x^\pr,y^\pr)$ be an arbitrary instance of \pathknapsack and $k$ the number of vertices in the solution. We can assume without loss of generality that we know $k$ since there are only $n-1$ possible values of $k$ namely $2,3,\ldots,n$ --- we can simply run the algorithm starting from $k=1$, incrementing $k$ in every iteration, till we obtain a solution which we output. We color every vertex uniformly randomly from a palette of $k$ colors independent of everything else. Let $\chi:\VV\longrightarrow[k]$ be the resulting coloring. For every non-empty subset $S\subseteq[k],S\ne\emptyset$ and vertex $x\in\VV$, we define a boolean variable PATH$(S,x)$ to be \true if there is a path which starts from $x^\pr$, ends at $x$, and contains exactly one vertex of every color in $S$; we call such a path $S$-colorful. If PATH$(S,x)$ is \true, then we also define $D[S,x]=\{(w,\alpha):\exists\text{ an $x^\pr$ to $x$ $S$-colorful path \PP such that }w(\PP)=w, \alpha(\PP)=\alpha\text{ for every other $x^\pr$ to $x$ $S$-colorful path \QQ, we have either }w(\QQ)>w \text{ or }\alpha(\QQ)<\alpha\}$. For $|S|=1$, we note that PATH$(S,x^\pr)$=\true if and only if $\{\chi(x^\pr)\}=S$; PATH$(S,x)$=\false for every $x\in\VV\setminus\{x^\pr\}$ and $|S|=1$; $D[S,x^\pr]=\{(w(x^\pr),\alpha(x^\pr))\}$ if $\{\chi(x^\pr)\}=S$; $D[S,x]=\emptyset$ for every $x\in\VV\setminus\{x^\pr\}$ and $S\subseteq\VV$. We update $\text{PATH}(S,x)$ as per the following recurrence for $|S|>1$.
    \[
    \text{PATH}(S,x) = 
    \begin{cases}
        \bigvee\{\text{PATH}(S\setminus\{\chi(x)\},v): \{v,x\}\in\EE\} & \text{if } \chi(x)\in S\\
        \false & \text{otherwise}
    \end{cases}
    \]
    When we update any $\text{PATH}(S,x)$ to be \true, we update $D[S,x]$ as follows. For every $\{x,y\}\in\EE$ if $\text{PATH}(S\setminus\{\chi(x)\},y)$ is \true, then we do the following: for every $(w,\alpha)\in D[S\setminus\{\chi(x)\},y]$, we put $(w+w(x),\alpha+\alpha(x))$ in $D[S,x]$ if $w+w(x)\le s$; and we finally remove all dominated pairs from $D[S,x]$. We output \yes if $\text{PATH}([k],y^\pr)$ is \true and there exists a $(w,\alpha)\in D[[k],y^\pr]$ such that $w\le s$ and $\alpha\ge d$. Otherwise, we output \no.

    {\bf Proof of correctness:} If there does not exist any colorful path between $x$ and $y$ of length $k$, then the algorithm clearly outputs \no. Suppose now that the instance is a \yes instance. Then there exists a colorful path $\PP=(u_1(=x),u_2,\ldots,u_k(=y))$ between $x$ and $y$ such that $w(\PP)\le s$ and $\alpha(\PP)\ge d$. Let us define $S_i=\{u_j:j\in[i]\}$ for $i\in[k]$. Then $\text{PATH}(S_i,u_i)$ is \true and $D[S_i,u_i]$ either contains $(w(S_i),\alpha(S_i))$ or any pair which dominates $(w(S_i),\alpha(S_i))$ for every $i\in[k]$. Hence, the algorithm outputs \yes.

    {\bf Running time analysis:} If there exists a colorful path between $x$ and $y$, then the algorithm finds it in time $\OO\left(2^k n^{\OO(1)}\right)$. If there exists a path between $x$ and $y$ containing $k-1$ edges (that is, $k$ vertices including $x$ and $y$), then one such path becomes colorful in the random coloring with probability at least
    \[ \frac{k!}{k^k}\ge e^{-k}. \]
    Hence, by repeating $\OO(e^k)$ times and outputting \yes if any run of the algorithm outputs \yes, the above algorithm achieves a success probability of at least $2/3$. We can use the splitters to derandomize the algorithm above to obtain a deterministic algorithm for \pathknapsack which runs in time $\OO\left((2e)^k k^{\OO(\log k)}n^{\OO(1)}\right)$~\cite[Section 5.6.2]{DBLP:books/sp/CyganFKLMPPS15}.
\end{proof}

\Cref{thm:sol-path} immediately implies the following result.

\begin{corollary}\label{cor:vcs-path}
    There is an algorithm for \pathknapsack running in time $\OO\left((2e)^{2vcs} vcs^{\OO(\log vcs)}n^{\OO(1)}\right)$ where $vcs$ is the size of the minimum vertex cover of the subgraph induced by the solution.
\end{corollary}
\section{\shortestpathknapsack}

We now consider \shortestpathknapsack. We observe that all the $u$ to $v$ paths in the reduced instance of \pathknapsack in the proof of \Cref{thm:pk-pathwidth} are of the same length. Hence, we immediately obtain the following result as a corollary of \Cref{thm:pk-pathwidth}.

\begin{corollary}\label{cor:shortest-pathwidth}
\shortestpathknapsack is \NPC even for graphs of pathwidth at most two and the weight of every edge is one. In particular, \shortestpathknapsack is \PNPH parameterized by pathwidth.
\end{corollary}

Interestingly, \shortestpathknapsack admits a pseudo-polynomial-time algorithm for any graph unlike \pa and \pathknapsack.

\begin{theorem}\label{thm:shortest-algo}
There is an algorithm for \shortestpathknapsack running in time $\OO((m+n\log n)\cdot\min\{s^2,(\alpha(\VV))^2\})$, where $m$ is the number of edges in the input graph.
\end{theorem}

\begin{proof}
Let $(\GG=(\VV,\EE,(c(e))_{e\in\EE}),(w(u))_{u\in\VV},(\alpha(u))_{u\in\VV},s,d,x,y)$ be an arbitrary instance of \shortestpathknapsack. We design a greedy and dynamic-programming based algorithm. For every vertex $v\in \VV$, we store a boolean marker $b_v$, the distance $\delta_v$ of $v$ from $x$, and a set $D_v=\{(w,\alpha):\exists$ an $x$ to $v$ shortest-path \PP such that $w(\PP)=w, \alpha(\PP)=\alpha,$ and for every other $x$ to $v$ shortest-path \QQ, we have either $w(\QQ)>w$ or $\alpha(\QQ)<\alpha \text{ (or both)}\}$. That is, we store undominated weight-value pairs of all shortest $x$ to $v$ paths in $D_v$. We initialize $b_x=\false, \delta_x=0, D_x=\{(w(x),\alpha(x))\}, b_u=\false, \delta_u=\infty,$ and $ D_u=\emptyset$ for every $u\in\VV\setminus\{x\}$.

Updating DP table: We pick a vertex $z=\argmin_{v\in\VV:b_v=\false}\delta_v$. We set $b_z=\true$. For every neighbor $u$ of $z$, if $\delta_u>\delta_z+c(\{z,u\})$, then we reset $D_u=\emptyset$ and set $\delta_u=\delta_z+c(\{z,u\})$. If $\delta_u=\delta_z+c(\{z,u\})$, then update $D_u$ as follows: for every $(w,\alpha)\in D_z$, we update 
$D_u$ to $(D_u\cup\{(w+w(u),\alpha+\alpha(u))\})$ if $w+w(u)\le s$. We remove all dominated pairs from $D_u$ just before finishing each iteration. If we have $b_v=\true$ for every vertex, then we output \yes if there exists a pair $(w,\alpha)\in D_y$ such that $w\le s$ and $\alpha\ge d$. Else, we output \no.

We now argue the correctness of our algorithm. Following the proof of correctness of the classical Dijkstra's shortest path algorithm, we observe that if $b_v$ is $\true$ for any vertex $v\in\VV$, its distance from $x$ is $\delta_v$~\cite{DBLP:books/daglib/0023376}. We claim that at the end of updating a table entry in every iteration, the following invariant holds: for every vertex $v\in\VV$ such that $b_v=\true$, $(k_1,k_2)\in D_v$ if and only if there exists an $x$ to $v$ undominated shortest path $\PP$ using only vertices marked $\true$ such that $w(\PP)=k_1$ and $\alpha(\PP)=k_2$.

The invariant clearly holds after the first iteration. Let us assume that the invariant holds after $i\;(>\!\!\!\!1)$ iterations; $\VV_T$ be the set of vertices which are marked $\true$ after $i$ iterations. We have $|\VV_T|=i>1$. Suppose the algorithm picks the vertex $z_{i+1}$ in the $(i+1)$-th iteration; that is, we have $z_{i+1}=\argmin_{v\in\VV:b_v=\false}\delta_v$ when we start the $(i+1)$-th iteration. Let $\PP^*=x,\ldots,z,z_{i+1}$ be an undominated $x$ to $z_{i+1}$ shortest path using the vertices marked $\true$ only. Then we need to show that $(w(\PP^*),\alpha(\PP^*))\in D_{z_{i+1}}$. We claim that the prefix of the path $\PP^*$ from $x$ to $z$, let us call it $\QQ=x,\ldots,z$, is an undominated $x$ to $z$ shortest path using the vertices marked \true only. It follows from the standard proof of correctness of Dijkstra's algorithm~\cite{DBLP:books/daglib/0023376} that $\QQ$ is a shortest path from $x$ to $z$. Now, to show that $\QQ$ is undominated, let us assume that another shortest path $\RR$ from $x$ to $z$ dominates $\QQ$. Then the shortest path $\RR^\pr$ from $x$ to $z_{i+1}$ which is $\RR$ followed by $z_{i+1}$ also dominates $\PP^*$ contradicting our assumption that $\PP^*$ is an undominated shortest path from $x$ to $z_{i+1}$. We now observe that the iteration $j$ when the vertex $z$ was marked \true must be less than $(i+1)$ since $z$ is already marked \true in the $(i+1)$-th iteration. Now, applying induction hypothesis after the $j$-th iteration, we have $(w(\QQ),\alpha(\QQ))\in D_z$ and $(w(\PP^*),\alpha(\PP^*))\in D_{z_{i+1}}$. Also, at the end of the $j$-th iteration, the $\delta_z$ and $\delta_{z_{i+1}}$ values are set to the distances of $z$ and $z_{i+1}$ from $x$ respectively. Thus, the $D_z$ and $D_{z_{i+1}}$ are never reset to $\emptyset$ after $j$-th iteration. Also, we never remove any undominated pairs from DP tables. Since $\PP^*$ is an undominated $x$ to $z_{i+1}$ path, we always have $(w(\PP^*),\alpha(\PP^*))\in D_{z_{i+1}}$ from the end of $j$-th iteration and thus, in particular, after $(i+1)$-th iteration. Hence, invariant (i) holds for $z_{i+1}$ after $(i+1)$-th iteration. Now consider any vertex $z^\pr$ other than $z_{i+1}$ which is marked \true before $(i+1)$-th iteration. Let $\PP_1=x,\ldots,z^\pr$ be an undominated $x$ to $z^\pr$ shortest path using the vertices marked \true only. If $\PP_1$ does not pass through the vertex $z_{i+1}$, then we have $(w(\PP_1),\alpha(\PP_1))\in D_{z_{z^\pr}}$ by induction hypothesis after $i$-iterations. If $\PP_1$ passes through $z_{i+1}$, we have $\delta_{z_{i+1}}<\delta_{z^\pr}$ since all the edge weights are positive. However, this contradicts our assumption that $z^\pr$ had already been picked by the algorithm and marked \true before $z_{i+1}$ was picked. Hence, $\PP_1$ cannot use $z_{i+1}$ as an intermediate vertex. Hence, after $(i+1)$ iterations, for every vertex $v\in\VV$ such that $b_v=\true$, for every $x$ to $v$ shortest path \PP using only vertices marked \true, we have $(w(\PP),\alpha(\PP))\in D_v$.

For the other direction, let $(k_1,k_2)$ be an arbitrary pair in $D_{z_{i+1}}$ after $(i+1)$ iterations. Let $j$ be the iteration when $(k_1,k_2)$ was first included in $D_{z_{i+1}}$; $z^\pr$ be the vertex that the algorithm picked and marked \true in the $j$-th iteration. Since $(k_1,k_2)$ was included in $D_{z_{i+1}}$ in the $j$-th iteration, we must have an edge between $z^\pr$ and $z_{i+1}$. From the proof of correctness of Dijkstra's algorithm, we observe that, after the end of the $j$-th iteration, $\delta_{z_{i+1}}$ is $\delta_{z^\pr}+c(\{z^\pr,z_{i+1}\})$ which is actually the distance of $z_{i+1}$ from $x$ and was thus never decreased after $j$-th iteration of the algorithm. Hence, $D_{z_{i+1}}$ is never reset to $\emptyset$ after the $j$-th iteration and thus the pair $(k_1,k_2)$ remains in $D_{z_{i+1}}$ from $j$-th iteration till $(i+1)$-th iteration. Also, it follows from the proof of correctness of Dijkstra's algorithm that there is a shortest path from $x$ to $z_{i+1}$ using vertices marked \true in the first $j$ iterations only. In particular, there exists an undominated shortest path $\PP_1$ from $x$ to $z_{i+1}$ using vertices marked \true in the first $(i+1)$ iterations only with $w(\PP_1)=k_1$ and $\alpha(\PP_1)=k_2$. Hence, the invariant holds for $z_{i+1}$ after $(i+1)$ iterations. For every other vertex $z^\prr\in\VV$ such that $z^\prr$ is marked \true, we have already argued before that we cannot have a shortest $x$ to $z^\prr$ path using $z_{i+1}$ as an intermediate vertex. Thus, the invariant holds for $z^\prr$ thanks to induction hypothesis.

For every vertex $v\in\VV$, the cardinality of $D_v$ is at most $s$ and also at most $\alpha(\VV)$ and thus at most $\min\{s,\alpha(\VV)\}$. Implementing the above algorithm using a standard Fibonacci heap-based priority queue to find $\argmin_{v\in\VV:b_v=\false}\delta_v$ gives us a running time $\OO((m+n\log n)\cdot\min\{s^2,(\alpha(\VV))^2\})$ where $m$ is the number of edges in the graph.
\end{proof}

Clearly, \shortestpathknapsack admits a polynomial-time algorithm for trees since only one path exists between every two vertices.

\begin{observation}\label{obs:shortest-tree-poly}
\shortestpathknapsack is in \Pb for trees.
\end{observation}

Using the technique in \Cref{thm-fptas}, we use \Cref{thm:shortest-algo} as a black-box to obtain the following approximation algorithm. We again omit its proof due to its similarity with \Cref{thm-fptas}.

\begin{theorem}\label{thm-fptas-shortest-path}
There is a $\text{poly}(n,1/\eps)$ time, $(1-\eps)$ factor approximation algorithm for \shortestpathknapsack for optimizing the value of the solution.
\end{theorem}
\section{Conclusion}

We study the classical \kp problem with various graph theoretic constraints, namely connectedness, path and shortest path. We show that \pa and \pathknapsack are strongly \NPC whereas \shortestpathknapsack admits a pseudo-polynomial time algorithm. All the three problems admit FPTASes for bounded treewidth graphs; only \shortestpathknapsack admits an FPTAS for arbitrary graphs. It would be interesting to explore if meta-theorems can be proven in this general theme of knapsack on graphs. \longversion{For example, one could study the problem on graph classes with properties such as heredity, etc. }

% \newpage
\bibliographystyle{splncs04}
\bibliography{references}
\shortversion{\newpage

}

\end{document}